\documentclass{article}
\usepackage{amssymb, amsmath, amsthm}
\numberwithin{equation}{section}
\newtheorem{theorem}{Theorem}[section]

\newtheorem{lemma}{Lemma}[section]
\newtheorem{corollary}{Corollary}[section]
\newtheorem{remark}{Remark}[section]

\begin{document}
\title{Inhomogeneous relativistic Boltzmann equation near vacuum in the Robertson-Walker space-time}
\author{\hspace{-0cm}$^{1}$Etienne TAKOU; $^{2}$Fid\`ele L. CIAKE CIAKE\\
\small{\hspace{-3cm}$^{1}$D\'epartement de Math\'{e}matiques, Ecole Nationale Sup\'erieure Polytechnique},\\ \small{Universit\'e de Yaound\'e 1, BP 8390, Yaounde, Cameroun, e-mail: takoueti@yahoo.com} \\ \small{\hspace{-5.3cm}$^{2}$D\'epartement de Math\'{e}matiques,  Ecole Normale Sup\'erieure,}\\ \small{\hspace{-1cm} Universit\'e de Yaound\'e 1, BP 47 Yaound\'e, Cameroun, e-mail:			
flciake@gmail.com
 }}
\date{}
\maketitle
\begin{abstract}
In this paper, we consider the Cauchy problem for the relativistic Boltzmann equation with near vacuum initial data where the distribution function depends on the time, the position and the impulsion. The collision kernel considered here is for the hard potentials case and the background space-time in which the study is done is the Robertson-Walker space-time. Unique global (in time) mild solution is obtained in a suitable weighted space.

\end{abstract}
\section{Introduction}
One of the most important equations in relativistic kinetic theory of gas is the Boltzmann equation. The main interest of the Boltzmann equation is the description of the one-particle distribution function associated to the gas. This function is physically interpreted as the probability of the presence density of a particle in a given volume. This equation describes the time evolution of the system where collisions between particles can no longer be neglected. One should consider effects of collisions by introducing in the right hand side of  the Vlasov equation a term of collision called collision operator. In this work, we assume that the particles interact only via binary and elastic collisions. This occurs when the mean free time is much shorter than the characteristic length time associated with the system. Therefore, between collisions for the case of uncharged particles considered in this paper, the particles move along future directed time-like geodesic of the space-time  ($ \mathbb{R}^{4}, ds^{2} $). From the tangent bundle point of view, the gas particles follow segments of integral curves of the vector field which will be specified later.

We consider as background in this paper the  Robertson-Walker(RW) space-time ($ \mathbb{R}^{4}, ds^{2} $) where  the metric tensor $ ds^{2} $ with signature (--, +, +, +) can be written as:
\begin{equation}\label{eq:1.1}
ds^{2} = -dt^{2} + R^{2}(t)[(dx^{1})^{2} + (dx^{2})^{2} + (dx^{3})^{2}]
\end{equation}
in (\ref{eq:1.1}), $ R(t) $  is given and is called the cosmological expansion factor. In fact, the RW metric is an exact solution of Einstein's field equations of General Relativity; it describes a homogeneous, isotropic expanding or contracting universe. The general form of the metric follows from the geometric properties of homogeneity and isotropy; Einstein's field equations are only needed to derive the scale factor of the universe as a function of time.

We consider in this work  particles with the same rest-mass $m$ that can be rescaled to $ m = 1$. The particles are then required to move on the future sheet of the mass-shell whose equation is $ -(p^{0})^{2} + R^{2}(t)[(p^{1})^{2} + (p^{2})^{2} + (p^{3})^{2}] = -1 $.

The main difficulty while studying the Boltzmann equation lays in the collision kernel. Glassey derived the collision kernel in \cite{glassey1} for the Newtonian case, but for the relativistic case for a more detail description of the scattering kernel, we refer to \cite{csernai}.

Choquet-Bruhat, Y.\cite{choquet} defined the $ \mu-N $ regularity of the collision operator, which appears in the r.h.s of the Boltzmann equation. Under the condition of $ \mu-N $ regularity, several authors studied and proved some existence theorems for the relativistic Boltzmann equation.  Takou, E., Noutchegueme, N. and Dongo, D.\cite{dongo, takou1, takou2} proved some global results in the homogeneous cases under assumptions close to $ \mu-N $ regularity. Bancel, D.\cite{bancel1} also proved one local result under the $ \mu-N $ regularity condition. Glassey, R.\cite{glassey} proved a global existence theorem with near vacuum initial data in the Minkowski space-time.

The use of $ \mu-N $ regularity doesn't allow a very good physical description of the collision operator. In fact, this operator depends on several terms including the collision kernel, the relative momentum and the energy in the center of momentum. Furthermore,
one of the main terms in the collision kernel is the scattering kernel which measures interactions between particles.

In the Newtonian Boltzmann equation, scattering kernels are usually classified  into soft and hard potentials. This classification was originally adapted in the relativistic case by Dudy\'nski, M. and Ekiel-Jer$\dot{z}$ewska, M.\cite{dud} and recently reformulated by Strain, R. in \cite{strain}. This reformulation increases the importance and the interest of the relativistic Boltzmann equation. With this reformulation, Strain, R. \textit{op.cit.} proved one global result in the Minkowski space-time. Lee, H. and Rendall, A., D. proved in \cite{lee1} the positivity of a possible solution  of the coupled Einstein-Boltzmann equation and also proved global solution in certain homogeneous cases in \cite{lee2, lee3}.

In this paper, we consider the relativistic Boltzmann equation with scattering kernel as in \cite{strain}, in RW space-time as indicated earlier. More precisely, we consider a relativistic gas of massive, uncharged particles in the RW space-time. In such case we look at short-range interactions between particles, which are usually modeled by a scattering kernel called hard sphere satisfying (\ref{eq:2.9}). For more details about the relativistic hard sphere interactions and their physical motivations we refer to \cite{dud1} and the references therein. 


The paper is organized as follows: In section 2, we introduce a change of variable to write the mild form of Boltzmann equation, we also specify the functions spaces in which we will seek the solution and then we make the main assumptions of the paper. Some preliminary results are given in section 3, whereas section 4 is devoted to the existence theorem of the relativistic Boltzmann equation.

\section{The inhomogeneous Boltzmann equation and Functional spaces}
\subsection{The equation and collision operator}
 We recall that we consider as background the RW space-time where the metric tensor with signature (--, +, +, +) can be
written as:
\begin{equation}\label{eq:2.1}
ds^{2} = -dt^{2} + R^{2}(t)[(dx^{1})^{2} + (dx^{2})^{2} + (dx^{3})^{2}]
\end{equation}
in which $ R(t) $ is a strictly positive function of $t$.

Let's also recall the general form of the Boltzmann equation on the curve space-time
\begin{equation}\label{eq:2.2}
p^{\alpha}\frac{\partial f}{\partial x^{\alpha}} - \Gamma^{i}_{\alpha\beta}p^{\alpha}p^{\beta}\frac{\partial f}{\partial p^{i}} = \tilde{Q}(f, f).
\end{equation}
In (\ref{eq:2.2}), $ \Gamma^{i}_{\alpha\beta} $ denote the Christoffel symbols of the metric considered, $ \tilde{Q} $ is a non-linear operator called "collision operator" and it will be specified  in detail shortly.

 Greek indices will be assumed to run from $0$ to $3$, while  latin indices run from $1$ to $3$, in (\ref{eq:2.2}) unless otherwise specified. We adopt the Einstein summation convention $ a_{\alpha}b^{\alpha} = \sum a_{\alpha}b^{\alpha}$. Note that $p^{\alpha}=(p^{0},p^{1},p^{2},p^{3})$ and $p=(p^{1},p^{2},p^{3})$.

After some computations, the relativistic Boltzmann equation in the RW space-time can be written as follows
%
\begin{equation}\label{eq:2.3}
\partial_{t}f + \hat{p}.\nabla_{x} f - 2\frac{\dot{R}}{R}p.\nabla_{p}f = Q(f, f)
\end{equation}
where $ \hat{p} $ is defined by $ \hat{p} = \frac{p}{p^{0}} $.

Let's now give the precise form of the collision operator. In instantaneous, binary and elastic scheme due to Lichnerowicz and Chernikov \cite{u}, we consider that at a given position $ x $, two particles(or two beans of particles) of momenta $ p^{\alpha} $ and $ q^{\alpha} $ collide without destroying each other. The collision affecting only their momenta  that change after the collision. Let $ p'^{\alpha} $ and $ q'^{\alpha} $ be their momenta after the collision. By conservation of the energy-momentum principle, one has:
\begin{equation}\label{eq:2.4}
p^{\alpha} + q^{\alpha} = p'^{\alpha} + q'^{\alpha}.
\end{equation}
The collision operator $ Q $ is then defined 
by the relation

$ Q(f, g) = Q_{g}(f, g) - Q_{l}(f, g) $ where:
\begin{eqnarray}
Q_{g}(f, g)(t, x, p) &= \int_{\mathbb{R}^{3}}\int_{S^{2}}\frac{g\sqrt{s}}{p^{0}q^{0}}  \sigma(g, \omega) f(t, x, p')g(t, x, q')d\omega dq \label{eq:2.5}\\
Q_{l}(f, g)(t, x, p) &= \int_{\mathbb{R}^{3}}\int_{S^{2}}\frac{g\sqrt{s}}{p^{0}q^{0}}\sigma(g, \omega) f(t, x, p)g(t, x, q)d\omega dq \label{eq:2.6}
\end{eqnarray}
correspond to the gain term and the lost term respectively. For simplicity, we abbreviate $ f(t, x, p) $, $ f(t, x, q) $, $ f(t, x, p') $ and $ f(t, x, q') $ by $ f(p) $, $ f(q) $,  $ f(p') $ and $ f(q') $ respectively.
The quantity $ v_{\phi} = \frac{g\sqrt{s}}{p^{0}q^{0}} $ is called M{\o}ller velocity.

In this paper $(p,q)$ and $(p',q')$ are pre-collisional and post-collisional momentum respectively, satisfying (\ref{eq:2.13}).

The quantities $ g $ and $ s $ defined as follows:
\begin{equation}\label{eq:2.7}
s = -(p_{\alpha} + q_{\alpha})(p^{\alpha} + q^{\alpha}), \quad \quad g = \sqrt{(p_{\alpha} - q_{\alpha})(p^{\alpha} - q^{\alpha})}
\end{equation}
are called respectively the square of the energy in the "center of momentum" system $ p + q = 0 $ and $ g $ the relative momentum.

$ \sigma $ is called the differential cross-section or scattering kernel; it depends on the relative momentum and the scattering angle $ \theta $ defined by the relation (\ref{eq:3.3}). Note that the parameter $ \omega $ over the unit sphere and the scattering angle $ \theta $ are linked by (\ref{eq:3.61}); that is why it is just written as $ \sigma(g, \omega) $.  $ \sigma(g, \omega) $  measures interaction's effects between particles during the collision process. The scattering kernel in relativistic kinetic theory is classified into soft and hard potentials.

\begin{itemize}
\item[-] For soft potentials, one assumes that there exists $ \gamma > -2 $ and \\ $ 0 < b < min\{4, 4 + \gamma\} $ such that the scattering kernel $ \sigma(g, \omega) $ satisfies the following growth/decay estimates:
\begin{equation}\label{eq:2.8}
\frac{g}{\sqrt{s}}g^{-b}\sigma_{0}(\omega) \lesssim \sigma(g, \omega) \lesssim g^{-b}\sigma_{0}(\omega), \quad \sigma_{0}(\omega) \lesssim sin^{\gamma}\theta.
\end{equation}
\item[-] For hard potentials, one assumes that there exists $ \gamma > -2 $, $ 0 \leq a \leq \gamma + 2 $  and $ 0 < b < min\{4, 4 + \gamma\} $ such that the scattering kernel $ \sigma(g, \omega) $ satisfies the following growth/decay estimates:
\begin{equation}\label{eq:2.9}
\frac{g}{\sqrt{s}}g^{a}\sigma_{0}(\omega) \lesssim \sigma(g, \omega) \lesssim (g^{a} + g^{-b})\sigma_{0}(\omega), \quad \sigma_{0}(\omega) \lesssim sin^{\gamma}\theta.
\end{equation}
\end{itemize}
The notation $ a \lesssim b $ means that a positive constant C exists  such that $ a \leq Cb $ holds uniformly over the range of parameters which are present in the inequality and moreover that the precise magnitude of the constant is unimportant. The notation $ a\approx b $  means that both $ a \lesssim b $  and $ b \lesssim a $ hold.

In the sequel, we let sometimes C and c denote  generic and positive inessential constants whose values may change from line to line.

\subsection{Hypothesis on the scattering kernel and the cosmological expansion factor}
 In the present work, we suppose that the scattering kernel $ \sigma(g, \omega) $ is for hard potentials case with $ a = 0 $. So, We assume that there exists  $ b \in ]0, 4[ $ such that the scattering kernel $ \sigma(g, \omega) $ satisfies the following growth/decay estimates:
\begin{equation}\label{eq:2.10}
\frac{g}{\sqrt{s}}\sigma_{0}(\omega) \lesssim \sigma(g, \omega) \lesssim (1 + g^{-b})\sigma_{0}(\omega)
\end{equation}
where $ \sigma_{0}(\omega) $ is non-negative, bounded, continuous and satisfies the following relation
\begin{equation}\label{eq:2.11}
\int_{S^{2}}\sigma_{0}(\omega)e^{-|w.y|^{2}} \lesssim e^{-|y|^{2}}, \quad \forall y \in \mathbb{R}^{3}\,\, \text{ such that} \,\, \, |y| \geq 1.
\end{equation}
About the cosmological expansion factor, we also assume that
\begin{equation}\label{eq:2.12}
\hspace{-0.01cm}R(0) = 1, \  R'(t) > 0, \  \underset{n \rightarrow +\infty}{lim}\!R(t) = +\infty, \  \int_{\mathbb{R}_{+}}\!\! (R^{-3}(t) + R^{b-4}(t))dt < + \infty.
\end{equation}
\begin{remark}  A scattering kernel enjoying (\ref{eq:2.10})-(\ref{eq:2.11}) falls into the hard potential case.
\end{remark}
\begin{remark} In \cite{glassey2}-Section 4, 
the post-collisional momenta were parametrized as follows: suppose that two particles having momenta $ V^{\alpha} $ and $ U^{\alpha} $ collide, and let $ V'^{\alpha} $ and $ U'^{\alpha} $ be their momenta after the collision. Under the energy-momentum conservation principle $ V^{\alpha} + U^{\alpha} =  V'^{\alpha} + U'^{\alpha}$, the following relations hold for $ \omega \in S^{2} $.
\begin{equation}\label{eq:2.120}
\begin{cases}
V' = V - A (V, U, \omega )\omega\\
U' = U + A (V, U, \omega )\omega \,
\end{cases}  with \, A  = \frac{2U^{0}V^{0}(V^{0} + U^{0})\omega.(\frac{V}{V^{0}} - \frac{U}{U^{0}})}{(V^{0} + U^{0})^{2} - (\omega.(V + U))^{2}}.
\end{equation}
\end{remark}
\begin{remark}
In (\ref{eq:2.120}), if we set $ V = Rp $, $ U = Rq $, $ V' = Rp' $ and $ U' = Rq' $, from (\ref{eq:2.4}), we have $ V^{\alpha} + U^{\alpha} =  V'^{\alpha} + U'^{\alpha}$. Then (\ref{eq:2.120}) holds for $U$ and $V$. from this we obtain the following relation between $(p, q)$ and $ (p', q') $
\begin{equation}\label{eq:2.13}
\begin{cases}
p' = p - \tilde{a} (p, q, \omega )\omega\\
q' = q + \tilde{a}(p, q, \omega )\omega \, ; \qquad
\omega \in S^{2}
\end{cases}
\end{equation}
in which, setting $ e = p^{0} + q^{0} $,  $\tilde{a}(p, q, \omega )$ is a real-valued function given by:
\begin{equation} \label{eq:2.14}
\tilde{a}(p, q, \omega ) = \frac{2\,p^{0}q^{0} e\,
\omega.(\hat{p} - \hat{q})}{e^{2} -
R^{2}(\omega.(p + q))^{2}}.
\end{equation}
As parametrization of the post-collisional momenta, we adopt (\ref{eq:2.13})-(\ref{eq:2.14}).
\end{remark}

\subsection{Mild form of the Boltzmann equation and functional space}
In the sequel, we consider (\ref{eq:2.3}) with covariant variables. To be explicit, the distribution function $ f $ will be considered as a function of $ t $, $ x $ and $ p_{k} = g_{k\beta}p^{\beta} = R^{2}p^{k} $, with $ k = 1, 2, 3 $. This change of variable was previously used  in \cite{lee2, lee3}.
In what follows, for simplicity, we set
\begin{equation}\label{eq:2.15}
v = (v^{1}, v^{2}, v^{3}) \quad where \quad v^{k} = R^{2}p^{k} \quad and \quad v^{0} = \sqrt{1 + R^{-2}|v|^{2}}.
\end{equation}

With these new variables, setting $ v'^{k} = R^{2}p'^{k} $ and $ u'^{k} = R^{2}q'^{k} $, the post-collisional momentum are parametrized as follows.
\begin{equation}\label{eq:2.16}
\begin{cases}
v' = v - a(v, u, \omega )\omega\\
u' =u + a(v, u, \omega )\omega \, ; \qquad
\omega \in S^{2}
\end{cases}
\end{equation}
where setting $ \hat{v} = \frac{v}{v^{0}} $ and $ \hat{u} = \frac{u}{u^{0}} $, the real valued function $ a $  is given by
\begin{equation}\label{eq:2.17}
a(v, u, \omega ) = \frac{2\,v^{0}u^{0} e\,
\omega.(\hat{v} - \hat{u})}{e^{2} -
R^{-2}(\omega.(v + u))^{2}}.
\end{equation}
With these variables we can now rewrite (\ref{eq:2.3}). Let's set\\ $\tilde{f}(t, x, v) = f(t, x, p)$. We have
\begin{equation}\label{eq:2.18}
\partial_{t}\tilde{f} = \partial_{t} f - 2\frac{\dot{R}}{R^{3}}v.\nabla_{p}f = \partial_{t} f - 2\frac{\dot{R}}{R}p.\nabla_{p}f
\end{equation}
\begin{equation}\label{eq:2.19}
\partial_{x_{i}}\tilde{f} = \partial_{x_{i}}{f}
\end{equation}
Straightforward computation leads to $ dp = R^{-6}dv $. In what follows, we will write $ f $ instead of $ \tilde{f} $. So, with the new variables, the collision operator reads
\begin{align}\label{eq:2.20}
Q(f, f)(t, x, v) &= R^{-3}(t)\int_{S^{2}}d\omega \int_{\mathbb{R}^{3}}du v_{\phi}\sigma(g, \omega)[f(v')f(u') - f(v)f(u)] \nonumber \\
&= Q_{g}(f, f)(t, x, v) - Q_{l}(f, f)(t, x, v).
\end{align}
Taking into account (\ref{eq:2.18}) and (\ref{eq:2.19}), the Boltzmann equation (\ref{eq:2.3}) becomes
\begin{equation}\label{eq:2.21}
\partial_{t}f + \frac{1}{R^{2}}\hat{v}.\nabla_{x}f =  Q_{g}(f, f)(t, x, v) - Q_{l}(f, f)(t, x, v).
\end{equation}
\subsubsection{Characteristic's equation}
Let's consider the equation (\ref{eq:2.21}) which is the first order partial differential equation. For any fixed $ (x, v) \in \mathbb{R}_{x}\times \mathbb{R}_{v} $, the characteristics $ X^{t}(x, v) $ are defined by the following relations
\begin{equation}\label{eq:2.22}
\frac{d}{dt} X^{t}(x, v) = R^{-2}(t)\hat{v}
\end{equation}
\begin{equation}\label{eq:2.23}
X^{t}(x, v)|_{t = 0} = x.
\end{equation}
From (\ref{eq:2.22}) and (\ref{eq:2.23}), we have
\begin{equation}\label{eq:2.24}
X^{t}(x, v) = x + \int_{0}^{t}R^{-2}(s) \hat{v}ds = x + \left(\int_{0}^{t}\frac{R^{-2}(s)ds}{\sqrt{1 + R^{-2}(s)|v|^{2}}}\right)v.
\end{equation}
Let's now introduce the standard notation in the Boltzmann equation
\begin{equation}\label{eq:2.25}
f^{\#}(t, x, v) = f(t, X^{t}(x, v), v)
\end{equation}
Using the notation (\ref{eq:2.25}), we have
\begin{align}\label{eq:2.26}
\frac{d}{dt}f^{\#}(t, x, v) &= \partial_{t}f +    \frac{\partial X^{i t} }{\partial t}\frac{\partial f}{\partial x^{i}}\nonumber\\
&= \partial_{t}f  + \frac{R^{-2}(t)}{\sqrt{1 + R^{-2}(t)|v|^{2}}}v.\nabla_{x}f\nonumber\\
&= \partial_{t}f  + R^{-2}(t)\hat{v}.\nabla_{x}f.
\end{align}
From (\ref{eq:2.26}), the equation (\ref{eq:2.21}) becomes
\begin{equation}\label{eq:2.27}
\frac{d}{dt}f^{\#}(t, x, v) = Q^{\#}(f, f)(t, x, v)
\end{equation}
where $ Q^{\#}(f,f) $ is given by:
\begin{equation*}
Q^{\#}(f, f)(s, x, v) = Q(f, f)(s, X^{s}(x, v), v).
\end{equation*}
(\ref{eq:2.27}) leads to the following equation
\begin{equation}\label{eq:2.28}
f^{\#}(t, x, v) =f_{0}(x, v) + \int_{0}^{t} Q^{\#}(f, f)(s, x, v)ds
\end{equation}
(\ref{eq:2.28}) is called the mild form of the Boltzmann equation. In what follows, we will focus on (\ref{eq:2.28}). 
\subsubsection{Functional space}
In the integral form of (\ref{eq:2.28}) for which we now look for a continuous bounded non-negative solution, we allow $ f $ to decay exponentially in $ v $ and $ x $. For this reason, we consider the weight function $ \rho $ defined by
\begin{equation}\label{eq:2.29}
\rho(x, v) = e^{(|v|^{2} + |x\times v|^{2})}.
\end{equation}
The function space in which we will seek the solution is defined as
\begin{equation}\label{eq:2.30}
M = \{f \in \mathcal{C}^{0}([0, +\infty[\times \mathbb{R}^{3}_{x}\times \mathbb{R}^{3}_{v}), \|f\| := \underset{t, x, v}{Sup}[\rho(x, v)|f(t, x, v)|] < +\infty\}.
\end{equation}

We can now state our main result.

\begin{theorem}
Define the operator $ \Gamma $ on M by
\begin{equation}
 \Gamma f^{\#} = f_{0}(x, v) + \int_{0}^{t}Q^{\#}(f, f)(\tau, x, v)d\tau
\end{equation}\label{eq:4.25}
 and let $ M_{r} = \{f \in M, \|f^{\#}\| \leq r\} $, under the assumptions (\ref{eq:2.10})-(\ref{eq:2.11}) on the collision kernel and (\ref{eq:2.12}) on the cosmological expansion factor, there exists a constant $ r_{0} $ such that if $ \|f_{0}\| $  is sufficiently small, the  integral equation $ \Gamma f^{\#} = f^{\#} $ has a unique solution $ f^{\#} \in M_{r_{0}} $.
\end{theorem}

Before giving the proof of our main result, we are going to collect some fundamental estimates.
\section{Preliminaries results}
\begin{lemma} The relative momentum enjoys the following estimates:
\begin{equation}\label{eq:3.1}
g \leq 2\sqrt{p^{0}q^{0}}, \quad and \quad \frac{R^{4}|p\times q|^{2} + R^{2}|p - q|^{2}}{p^{0}q^{0}} \leq g^{2} \leq R^{2}|p - q|^{2}.
\end{equation}
\end{lemma}
\begin{proof}
It's obvious to prove the relation $ s = g^{2} + 4 $. As a consequence, $ s \geq 4 $. On the another hand, for $ s $, we have:
\begin{equation*}
s = -p^{\alpha}p_{\alpha} -q^{\alpha}q_{\alpha} - 2p^{\alpha}q_{\alpha} = 2 + 2p^{0}q^{0} - 2g_{ij}p^{i}q^{j} = 2p^{0}q^{0} + 2- 2R^{2}p.q.
\end{equation*}
Since $ 1 - R^{2}p.q \leq \sqrt{1 + R^{2}|p|^{2} + R^{2}|q|^{2} + R^{4}|p|^{2}|q|^{2}} = p^{0}q^{0}$, it follows that $ s \leq 4p^{0}q^{0} $ and then $ g=\sqrt{s-4} \leq 2\sqrt{p^{0}q^{0}} $.

- For proving the first part of the second inequality, we use the elementary estimate $ 1 + R^{2}p.q \leq p^{0}q^{0}$. We then have

 \begin{align*}
g^{2} &= -2 + 2p^{0}q^{0} - 2R^{2}p.q\\
&= 2\frac{(p^{0}q^{0})^{2}  - (1 + R^{2}p.q)^{2}}{p^{0}q^{0} + 1 + R^{2}p.q}
\\&= 2\frac{(1 + R^{2}|p|^{2})(1 + R^{2}|q|^{2}) -1  - R^{4}(p.q)^{2} - 2R^{2}p.q}{p^{0}q^{0} + 1 + R^{2}p.q}
\\&= 2\frac{R^{2}|p|^{2} + R^{2}|q|^{2} -1  - R^{4}(|p|^{2}|q|^{2} - (p.q)^{2}) - 2R^{2}p.q}{p^{0}q^{0} + 1 + R^{2}p.q}\\&= 2\frac{R^{4}|p\times q|^{2} + R^{2}|p - q|^{2}}{p^{0}q^{0} + 1 + R^{2}p.q}\\
&\geq 2\frac{R^{4}|p\times q|^{2} + R^{2}|p - q|^{2}}{2p^{0}q^{0}}.
\end{align*}

- About the last inequality, let $ \theta_{0} $ be the angle between $ p - q $ and $ p + q $. We have $ (p^{0})^{2} - (q^{0})^{2} = R^{2}|p - q||p + q|cos\theta_{0} $. On the other hand
\begin{align*}
g^{2} &= -(p^{0} - q^{0})^{2} + (p^{0})^{2} + (q^{0})^{2} - 2(1 + R^{2}p.q)\\
&=  -(p^{0} - q^{0})^{2} + R^{2}(|p|^{2} + |q|^{2} - 2p.q)\\
&= R^{2}|p - q|^{2} - R^{4}\left[\frac{(p - q)(p + q)cos\theta_{0} }{p^{0} + q^{0}}\right]^{2}\\
&= R^{2}|p - q|^{2}\left[1 - \frac{R^{2}|p + q|^{2}cos^{2}\theta_{0}}{(p^{0} + q^{0})^{2}} \right]\\
&\leq R^{2}|p - q|^{2}.
\end{align*}
\end{proof}

\begin{lemma}
\end{lemma} The function $ a(u, v, \omega) $ enjoys the estimates
\begin{equation}\label{eq:3.2}
\frac{2p^{0}q^{0}|\omega.(\hat{p} - \hat{q})|}{e} \leq |\tilde{a}(p, q, \omega)| \leq \frac{ e |p - q|}{\sqrt{e^{2} - R^{2}|p + q|^{2}}} = \frac{e|p - q|}{\sqrt{s}}.
\end{equation}
\begin{proof}

The lower bound is trivial since $ e^{2} - R^{2}(\omega.(p + q))^{2} \leq e^{2} $.

The scattering angle $ \theta $ such that
\begin{equation}\label{eq:3.3}
\cos \theta = \frac{(p^{\alpha} - q^{\alpha})(p'_{\alpha} - q'_{\alpha})}{g^{2}}
\end{equation}
is  well defined under the energy-momentum conservation principle; see \cite{glassey1}, Lemma 3.15.3. Let's compute the numerator $ N $ of $ \cos \theta $.
\begin{align*}
(p^{\alpha} - q^{\alpha})(p'_{\alpha} - q'_{\alpha}) &= -(p^{0} - q^{0})(p'^{0} - q'^{0}) + R^{2}(p - q)(p' - q')\\
&= -(p^{0} - q^{0})(p'^{0} - q'^{0}) + R^{2}(p - q)(p - q - 2 \tilde{a}\omega)
\end{align*}
The term $ p'_{0} - q'_{0} $ is given by
\begin{align*}
p'^{0} - q'^{0} &= \frac{(p'^0)^{2} - (q'^{0})^{2}}{e} = \frac{R^{2}(|p'|^{2} - |q'|^{2})}{e}\\
&= \frac{R^{2}(|p|^{2} - |q|^{2} - 2\tilde{a}\omega(p + q))}{e}\\
&= \frac{(p^{0})^{2} - (q^{0})^{2} - 2R^{2}\tilde{a} \omega.(p + q)}{e}
\end{align*}

\begin{align*}
(p^{0} - q^{0})(p'^{0} - q'^{0}) &= \frac{(p^{0} - q^{0})[(p^0)^{2} - (q^{0})^{2} - 2R^{2}\tilde{a}\omega.(p + q)]}{e} \\&= \frac{(p^{0} + q^{0})(p^0 - q^{0})^{2} - 2R^{2}\tilde{a}(p^{0} - q^{0})\omega.(p + q)}{e} .
\end{align*}
The numerator of $ \cos \theta $ becomes

\begin{align*}
N &= -(p^{0} - q^{0})^{2} + \frac{2R^{2}\tilde{a}(p^{0} - q^{0})\omega.(p + q)}{e} + R^{2}|p - q|^{2} - 2R^{2}\tilde{a}\omega.(p - q)\\
&= g^{2} -  \frac{2R^{2}\tilde{a}\omega}{e}[(p^{0} + q^{0})(p - q) - (p^{0} - q^{0})(p + q)]\\
&= g^{2} - \frac{2R^{2}\tilde{a}\omega}{e}[2q^{0}p - 2p^{0}q]\\
&= g^{2} - \frac{4p^{0}q^{0}R^{2}\tilde{a}\omega.(\hat{p} - \hat{q})}{e}\\
&= g^{2} - 2\tilde{a}R^{2}\frac{2p^{0}q^{0}e\omega.(\hat{p} - \hat{q})}{e^{2} - R^{2}(\omega.(p + q))^{2}}\times \frac{e^{2} - R^{2}(\omega.(p + q))^{2}}{e^{2}}\\
&= g^{2} - 2 R^{2}\tilde{a}^{2}\frac{e^{2} - R^{2}(\omega.(p + q))^{2}}{e^{2}}.
\end{align*}
From the relation $ \cos \theta = 1 - 2\sin^{2}\frac{\theta}{2} $, and the relation $ |\omega.(p + q)| \leq |p + q| $  it follows that
\begin{align*}
\sin^{2}\frac{\theta}{2} = R^{2}\tilde{a}^{2}\frac{e^{2} - R^{2}(\omega.(p + q))^{2}}{g^{2}e^{2}} \Rightarrow |\tilde{a}(p, q, \omega)| &\leq \frac{R^{-1}ge}{\sqrt{e^{2} - R^{2}|p + q|^{2}}}\\
&\leq  \frac{e|p - q|}{\sqrt{s}}.
\end{align*}
\end{proof}
\begin{corollary}
 The function $a(v, u, \omega)$ and $ \omega.(\hat{v} - \hat{u}) $ enjoy the following estimates:
\begin{equation}\label{eq:3.4}
|a(v, u, \omega)| \leq \frac{R^{-2} e|v - u|}{\sqrt{e^{2} - R^{-2}|v + u|^{2}}},
\end{equation}
\begin{equation}\label{eq:3.5}
|\omega.(\hat{v} - \hat{u})| \lesssim \frac{Rge}{v^{0}u^{0}} \lesssim \frac{e|v - u|}{v^{0}u^{0}}.
\end{equation}
\end{corollary}
\begin{proof}

- (\ref{eq:3.4}) is a direct consequence of (\ref{eq:2.14}), (\ref{eq:2.17}) and (\ref{eq:3.2}).

- About (\ref{eq:3.5}), in the expression above of $ \sin^{2}\frac{\theta}{2} $, we replace $ \tilde{a}(p, q, \omega) $ by its expression (\ref{eq:2.14}). We then obtain
\begin{align}\label{eq:3.61}
\sin^{2}\frac{\theta}{2} &= R^{2}\left(\frac{2p^{0}q^{0}e \omega.(\hat{p} - \hat{q})}{e^{2} - R^{2}(\omega.(p + q))^{2}}\right)^{2}\frac{e^{2} - R^{2}(\omega.(p + q))^{2}}{g^{2}e^{2}}\nonumber\\
& = \frac{4R^{2}(p^{0}q^{0})^{2}(\omega.(\hat{p} - \hat{q}))^{2}}{g^{2}(e^{2} - R^{2}(\omega.(p + q))^{2})}.
\end{align}
This is the precise relationship between the parameter $ \omega $ over the sphere and scattering angle $ \theta $.
We then deduce that
\begin{equation*}
|\omega.(\hat{p} - \hat{q})| = R^{-1}\frac{g\sin \frac{\theta}{2}\sqrt{e^{2} - R^{2}(\omega.(p + q))^{2}}}{2p^{0}	q^{0}}.
\end{equation*}
Returning to the variables $ v $ and $ u $, this leads to
\begin{equation*}
|\omega.(\hat{v} - \hat{u})| = \frac{Rg\sin \frac{\theta}{2}\sqrt{e^{2} - R^{-2}(\omega.(v + u))^{2}}}{2v^{0}	u^{0}} \leq \frac{Rge}{v^{0}u^{0}} \leq \frac{e|v - u|}{v^{0}u^{0}}.
\end{equation*}
\end{proof}

\begin{lemma} Given a positive constant $ B $,  for fixed  $ v $ and $ u $, there exists $ t_{0} \in \mathbb{R}_{+} $ such that in $ [t_{0}, +\infty[ $, $\Omega(t) = a\omega.(v - u)  $ is bounded from above by B.
\end{lemma}

\begin{proof} Direct computation leads to
\begin{equation}\label{eq:3.6}
a\omega.(v - u) = av^{0}\omega.(\hat{v} - \hat{u}) + a(v^{0} - u^{0})\omega.\hat{u}.
\end{equation}
From (\ref{eq:3.4})-(\ref{eq:3.5}), the first term in the right hand side of (\ref{eq:3.6}) is controlled as follows
\begin{equation*}
|av^{0}\omega.(\hat{v} - \hat{u})| \leq \frac{R^{-2}e^{2}|v - u|^{2}}{2\sqrt{e^{2} - R^{-2}|v + u|^{2}}} = \frac{R^{-2}e^{2}|v - u|^{2}}{2\sqrt{s}}.
\end{equation*}
Let's recall that $ e = \sqrt{1 + \frac{|v|^{2}}{R^{2}}} + \sqrt{1 +  \frac{|u|^{2}}{R^{2}}} $ and $\underset{t \rightarrow +\infty}{lim}R(t) = +\infty$. So, $ e^{2} $ goes to $ 4 $ as $ t $ goes to $ +\infty $. Since $ \sqrt{s} \geq 2 $, $ av^{0}\omega.(\hat{v} - \hat{u}) $ tends to zero as $ t $ goes to $ +\infty $.

About the second term in the right hand side of (\ref{eq:3.6}), let's observe that $ |a(v^{0} - u^{0})\omega.\hat{u}| \leq |a(v^{0} - u^{0})||u| $. The relation (\ref{eq:3.4}) together with the equality $ v^{0} - u^{0} = R^{-2}(|v|^{2} - |u|^{2}) $ and the fact that $ \underset{t \rightarrow + \infty}{lim} R(t) = + \infty $ allow us to claim that $ a(v^{0} - u^{0})\omega.\hat{u} $ goes to zero as $ t $ goes to $ +\infty $.

Thus, for a given positive constant B,  there exists $ t_{0} \in \mathbb{R}_{+} $ such that beyond $ t_{0} $, one has $\Omega(t) = a\omega.(v - u) \leq B $.

\end{proof}

 With the parametrization (\ref{eq:2.16}) we can prove the following inequality for $ \omega \in S^{2}_{+} $ where $ S^{2}_{+} =  \{\omega \in S^{2}, a\omega.(v - u) \leq B\} $ is a restriction of $ S^{2} $.

\begin{lemma}
Let v and u be given. suppose that v' and u' are parametrized as indicated in (\ref{eq:2.16}),(\ref{eq:2.17}) with an unit vector $ \omega \in S^{2}_{+}  $. Then, we have the following estimate.
\begin{equation}\label{eq:3.7}
|v|^{2} + |u|^{2} - |v'|^{2}- |u'|^{2} \leq B.
\end{equation}
\end{lemma}

\begin{proof} Straightforward computation leads to
\begin{equation*}
|v|^{2} + |u|^{2} - |v'|^{2}- |u'|^{2} = -2a^{2} + 2a\omega.(v - u) \leq 2a\omega.(v - u) \leq B.
\end{equation*}
\end{proof}

\begin{remark} The restriction of the type $ S^{2}_{+} $ on the set $ S^{2} $ was previously used by Strain, R. \cite{strain3} and Lee, H. \cite{lee3}. Note that  $S^{2}_{+} $ depends on $ v $, $ u $ and $ t $. From lemma 3.3,  we can find a finite $ t_{0} $ such that $ S^{2}_{+} = S^{2} $ for $ t \geq t_{0} $. This means that the restriction on $ S^{2} $ disappears for large $ t $. In the sequel, we consider the collision operator with the restriction $ S^{2}_{+} $.
\end{remark}

\begin{lemma}
Suppose that $ \sigma_{0}(\omega) $ satisfies the boundedness assumption (\ref{eq:2.10}), then we have the following inequalities:
\begin{equation}\label{eq:3.8}
\int_{\mathbb{R}^{3}}v_{\phi}g^{-b}e^{-|u|^{2}}du \leq C \quad for \quad 0 \leq b \leq 1,
\end{equation}
\begin{equation}\label{eq:3.9}
\int_{\mathbb{R}^{3}}v_{\phi}g^{-b}e^{-|u|^{2}}du \leq CR^{b - 1} \quad for \quad 1 \leq b < 4.
\end{equation}
\end{lemma}

\begin{proof} This lemma was proved in \cite{lee3} and we just present it for the reader convenience.
Let's recall that R satisfies assumptions (\ref{eq:2.12}). \\
- About the  inequality (\ref{eq:3.8}), we have
\begin{align*}
\int_{\mathbb{R}^{3}}v_{\phi}g^{-b}e^{-|u|^{2}}du =  \int_{\mathbb{R}^{3}}\frac{g^{1 - b}\sqrt{s}}{v^{0}u^{0}}e^{-|u|^{2}}du \leq C \int_{\mathbb{R}^{3}}(v^{0}u^{0})^{-\frac{b}{2}}e^{-|u|^{2}}du \leq C
\end{align*}
- Let's now prove the second inequality. We start with the case $ 1 \leq b \leq 2 $
\begin{align*}
\int_{\mathbb{R}^{3}}\!v_{\phi}g^{-b}e^{-|u|^{2}}du \!&= \!\! \int_{\mathbb{R}^{3}}\!\frac{g^{1 - b}\sqrt{s}}{v^{0}u^{0}}e^{-|u|^{2}}\!du \!\leq \! C\! \int_{\mathbb{R}^{3}}\!\frac{1}{\sqrt{v^{0}u^{0}}}\frac{R^{b-1}(v^{0}u^{0})^{\frac{b-1}{2}}}{|v - u|^{b - 1}} e^{-|u|^{2}}\!du\\
 &\leq C R^{b - 1}\int_{\mathbb{R}^{3}}\frac{1}{(v^{0}u^{0})^{\frac{2 - b}{2}}}\frac{1}{|v - u|^{b - 1}} e^{-|u|^{2}}du\\
  &\leq C R^{b - 1}\int_{\mathbb{R}^{3}}\frac{1}{|v - u|^{b - 1}} e^{-|u|^{2}}du\\
  &\leq C R^{b - 1}(1 + |v|^{2})^{\frac{1-b}{2}} \leq C R^{b - 1}
\end{align*}
where the last inequality is obtained by using the inequality\\ $ \int_{\mathbb{R}^{3}}|v - u|^{-\alpha}e^{-|u|^{2}} \leq C_{\alpha}(1 + |v|^{2})^{\frac{-\alpha}{2}} $ given in \cite{lee3}\\

- Now, we consider the case $ 2 \leq b < 4 $. We have
\begin{align*}
\int_{\mathbb{R}^{3}}\!\!v_{\phi}g^{-b}e^{-|u|^{2}}du \!&= \!\! \int_{\mathbb{R}^{3}}\frac{g^{1 - b}\sqrt{s}}{v^{0}u^{0}}e^{-|u|^{2}}\!du\!\! \leq\!\! C \!\!\int_{\mathbb{R}^{3}}\!\frac{1}{\sqrt{v^{0}u^{0}}}\frac{R^{b-1}(v^{0}u^{0})^{\frac{b-1}{2}}}{|v - u|^{b - 1}} e^{-|u|^{2}}\!du\\
 &\leq C R^{b - 1}\int_{\mathbb{R}^{3}}\frac{(v^{0}u^{0})^{\frac{b - 2}{2}}}{|v - u|^{b - 1}} e^{-|u|^{2}}du\\
 &\leq C R^{b - 1}\int_{\mathbb{R}^{3}}\frac{(1 + |v|^{2})^{\frac{b - 2}{4}}(1 + |u|^{2})^{\frac{b - 2}{4}}}{|v - u|^{b - 1}} e^{-|u|^{2}}du
 \\
 &\leq C R^{b - 1}(1 + |v|^{2})^{\frac{b - 2}{4} - \frac{b - 1}{2}}
  \\
  &\leq C R^{b - 1}(1 + |v|^{2})^{-\frac{b}{4}} \leq C R^{b - 1}.
\end{align*}
\end{proof}
With this preliminaries results in hand, we look for the existence theorem which is the main result of this paper.
\section{Estimates on the collision operator}
First of all, we will try to control the loss term and the gain term.

In order to control the loss term, let's observe that
\begin{equation}\label{eq:4.1}
\hspace{-0.01cm}Q_{l}^{\#}(f, f)(t, x, v)\! =\! R^{-3}(t)f^{\#}(t, x, v)\!\int_{S^{2}_{+}}\!\!d\omega\!\! \int_{\mathbb{R}^{3}}\!\!du v_{\phi}\sigma(g, \omega)f(t, X^{t}(x, v), u).
\end{equation}
We look for an element $ y \in \mathbb{R}^{3}_{x} $ satisfying the relation $ f(t, X^{t}(x, v), u) = f^{\#}(t, y, u) $. This holds  if $ y = x + b(t, u, v) $ where the function b is defined as:
\begin{equation}\label{eq:4.2}
b(t, u, v) = \int_{0}^{t}\left(\frac{R^{-2}(s)v}{\sqrt{1 + R^{-2}(s)|v|^{2}}} - \frac{R^{-2}(s)u}{\sqrt{1 + R^{-2}(s)|u|^{2}}}\right)ds.
\end{equation}

\begin{lemma}
Under hypotheses (\ref{eq:2.10}) and (\ref{eq:2.11}) on the collisional cross section $ \sigma(g, \omega) $ and the assumption (\ref{eq:2.12}) on the scalar factor $ R(t) $, for any $ t \geq 0 $ and $ f^{\#} \in M $, there is a constant c independent on $ t, x, v $ for which
\begin{equation}\label{eq:4.3}
\int_{0}^{t}|Q_{l}^{\#}(f, f)(\tau, x, v)|d\tau \leq c \rho(x, v)^{-1}\|f^{\#}\|^{2}.
\end{equation}
\end{lemma}
\begin{proof}. We have
\begin{align}\label{eq:4.4}
&\int_{0}^{t}|Q_{l}^{\#}(f, f)(\tau, x, v)|d\tau \nonumber\\
&= \int_{0}^{t}\!\left|R^{-3}(\tau)d\tau f^{\#}(\tau, x, v)\!\int_{S^{2}_{+}}\!d\omega\int_{\mathbb{R}^{3}}\!\!du \frac{g\sqrt{s}}{v^{0}u^{0}}f^{\#}(\tau, x \!+\! \int_{0}^{\tau}\!R^{-2}(s)(\hat{v} -\hat{u})ds, u)\right|\nonumber\\
&\leq \rho^{-1}(x, v)\|f^{\#}\|^{2}\int_{0}^{t}R^{-3}(\tau)d\tau \int_{S^{2}_{+}}\int_{\mathbb{R}^{3}} \dfrac{ v_{\phi}\sigma(g, \omega)dud\omega }{e^{|u|^{2} + |(x + \int_{0}^{\tau}R^{-2}(s)\hat{v}ds)\times u|^{2}}}\nonumber\\
&\leq \rho^{-1}(x, v)\|f^{\#}\|^{2}\int_{0}^{t}R^{-3}(\tau)d\tau \int_{S^{2}_{+}}\int_{\mathbb{R}^{3}}v_{\phi}\sigma(g, \omega)e^{-|u|^{2}}dud\omega.
\end{align}
Consider the term $I_{t} = \int_{S^{2}_{+}}\int_{\mathbb{R}^{3}}d\omega du v_{\phi}\sigma(g, \omega)e^{-|u|^{2}}  $. Using the fact that\\ $ \sqrt{s} \leq 2\sqrt{v^{0}u^{0}} $ and  $ g \leq 2\sqrt{v^{0}u^{0}} $,  we have
\begin{align*}
I_{t} &\leq \int_{S^{2}_{+}}\int_{\mathbb{R}^{3}} \frac{g\sqrt{s}}{v^{0}u^{0}}(1 + g^{-b})\sigma_{0}(\omega)e^{-|u|^{2}}du d\omega  \\
&\lesssim \int_{S^{2}_{+}}\int_{\mathbb{R}^{3}}\sigma_{0}(\omega)e^{-|u|^{2}}du d\omega  + \int_{S^{2}_{+}}\int_{\mathbb{R}^{3}} \frac{g\sqrt{s}}{v^{0}u^{0}}g^{-b}\sigma_{0}(\omega)e^{-|u|^{2}}dud\omega  \\
&\lesssim \int_{\mathbb{R}^{3}}
e^{-|u|^{2}}du  + \int_{\mathbb{R}^{3}} \frac{g\sqrt{s}}{v^{0}u^{0}}g^{-b}e^{-|u|^{2}}du.
\end{align*}
It follows that
\begin{equation}\label{eq:4.5}
I_{t} \leq C \,\, \text{ if }\,\, 0 \leq b \leq 1  \text{ and } I_{t} \leq C(1 + R^{b - 1})\,\text{ if } \,\, 1 \leq b \leq 3.
\end{equation}
Under the assumptions (\ref{eq:2.12}) stating that $ R^{-3} $ and $ R^{b - 4} $ are integrable over $ [0, +\infty[ $, we obtain the desired result.
\end{proof}
\begin{lemma}
Under hypotheses (\ref{eq:2.10}) and (\ref{eq:2.11}) on the collisional cross section $ \sigma(g, \omega) $ and the assumption (\ref{eq:2.12}) on the scalar factor $ R(t) $, for any $ t \geq 0 $ and $ f^{\#} \in M $, there is a constant c independent of $ t, x, v $ for which
\begin{equation}\label{eq:4.6}
\int_{0}^{t}|Q_{g}^{\#}(f, f)(\tau, x, v)|d\tau \leq c \rho(x, v)^{-1}\|f^{\#}\|^{2}.
\end{equation}
\end{lemma}

\begin{proof}
About the gain term, using the function $ b(t, u, v) $ defined in (\ref{eq:4.2}), it follows that
\begin{equation}\label{eq:4.7}
\begin{cases}
f(t, X^{t}(x, v), v') = f^{\#}(t, x + b(t, v', v), v') \\ f(t, X^{t}(x, v), u') = f^{\#}(t, x + b(t, u', v), u').
\end{cases}
\end{equation}
From (\ref{eq:4.7}), one has
\begin{align}\label{eq:4.8}
&\int_{0}^{t}|Q_{g}^{\#}(f, f)(\tau, x, v)|d\tau \nonumber\\
&= \int_{0}^{t}\left|R^{-3}(\tau)d\tau \int_{S^{2}_{+}}d\omega\int_{\mathbb{R}^{3}}du \frac{g\sqrt{s}}{v^{0}u^{0}}f(\tau,X^{\tau}(x, v), v')f(\tau,X^{\tau}(x, v), u'')\right|\nonumber\\
&=\! \!\int_{0}^{t}\!\!d\tau\!\bigg|\int_{S^{2}_{+}}\! \!\! \!d\omega\int_{\mathbb{R}^{3}}\! \!\! \!du \frac{R^{-3}(\tau)g\sqrt{s}}{v^{0}u^{0}}f^{\#}(\tau, x \!+\!\! \! \int_{0}^{\tau}\! \!\! \!R^{-2}(s)(\hat{v} \!-\!\hat{v'})ds, v')\nonumber\\
&\hspace{3cm}\times f^{\#}(\tau, x \! \!+\! \! \int_{0}^{\tau}\! \!R^{-2}(s)(\hat{v}\! -\!\hat{u'})ds, u')\bigg|\nonumber\\
\end{align}
\begin{align}
&\leq \!\!\|f^{\#}\|^{2}\!\!\int_{0}^{t}\!\!d\tau\!\! \int_{S^{2}_{+}}\!\!\int_{\mathbb{R}^{3}}\!\! \dfrac{R^{-3}(\tau)v_{\phi}\sigma(g, \omega)}{e^{ |v'|^{2} + |(x + \int_{0}^{\tau}R^{-2}(s)\hat{v}ds)\times v'|^{2}}}\nonumber\\
&\hspace{3cm}\times \dfrac{d\omega du }{e^{|u'|^{2} + |(x + \int_{0}^{\tau}R^{-2}(s)\hat{v}ds)\times u'|^{2}}}\nonumber\\
&\leq\!\! \|f^{\#}\|^{2}\!\!\int_{0}^{t}\!\!d\tau\!\! \int_{S^{2}_{+}}\!\!\int_{\mathbb{R}^{3}} \!\! \dfrac{R^{-3}(\tau) v_{\phi}\sigma(g, \omega)}{e^{ |v'|^{2} + |u'|^{2}}}\nonumber\\
&\hspace{3cm}\times\dfrac{d\omega du }{e^{|(x + \int_{0}^{\tau}R^{-2}(s)\hat{v}ds)\times v'|^{2} + |(x + \int_{0}^{\tau}R^{-2}(s)\hat{v}ds)\times u'|^{2}}}.
\end{align}

From (\ref{eq:3.7}), we have $ e^{|v'|^{2} + |u'|^{2}} \leq ce^{|v|^{2} + |u|^{2}} $. So (\ref{eq:4.8}) leads to
\begin{align}\label{eq:4.9}
&\int_{0}^{t}|Q_{g}^{\#}(f, f)(\tau, x, v)|d\tau \nonumber\\
&\leq c e^{-|v|^{2}}\|f^{\#}\|^{2}\!\!\int_{0}^{t}\!R^{-3}d\tau\!\! \int\!\! \dfrac{d\omega du v_{\phi}\sigma(g, \omega)e^{-|u|^{2}}}{e^{|(x + \int_{0}^{\tau}R^{-2}(s)\hat{v}ds)\times v'|^{2} + |(x + \int_{0}^{\tau}R^{-2}(s)\hat{v}ds)\times u'|^{2}}}.
\end{align}
We now try to control the term D defined by
\begin{equation}\label{eq:4.10}
D = |(x + \int_{0}^{\tau}R^{-2}(s)\hat{v}ds)\times v'|^{2} + |(x + \int_{0}^{\tau}R^{-2}(s)\hat{v}ds)\times u'|^{2}.
\end{equation}
Let's define the following vectors and scalars.
\begin{equation}\label{eq:4.11}
a_{v} = x\times v'; \quad b_{v} = v\times v'; \quad \nu_{v} = \frac{b_{v}}{|b_{v}|}; \quad c_{v} = a_{v}.b_{v},
\end{equation}
\begin{equation}\label{eq:4.12}
a_{u} = x\times u'; \quad b_{u} = v\times u'; \quad \nu_{u} = \frac{b_{u}}{|b_{u}|}; \quad c_{u} = a_{u}.b_{u}.
\end{equation}
If we set $ \chi(\tau) = \int_{0}^{\tau}\frac{R^{-2}(s)ds}{\sqrt{1 + R^{-2}(s)|v|^{2}}} $, we have
\begin{align*}
D &= (a_{u} + \chi(\tau)b_{u})^{2} + (a_{v} + \chi(\tau)b_{v}|)^{2}\\
& = (|b_{v}|^{2} + |b_{u}|^{2})\chi(\tau)^{2} + 2(c_{v} + c_{u})\chi(\tau) + (|a_{v}|^{2} + |a_{u}|^{2}).
\end{align*}
So, D is a polynomial of second order in $ \chi(\tau) $. Let's prove that the opposite of its discriminant $ \Delta $ is bounded from below. We have
\small{\begin{align*}
-\Delta &= (|b_{v}|^{2} + |b_{u}|^{2})(|a_{v}|^{2} + |a_{u}|^{2}) - (c_{v} + c_{u})^{2}\\
&= |b_{v}|^{2}|a_{v}|^{2} - (a_{v}.b_{v})^{2} + |b_{v}|^{2}|a_{u}|^{2}  + |b_{u}|^{2}|a_{u}|^{2}  - (a_{u}.b_{u})^{2} + |b_{u}|^{2}|a_{v}|^{2}  - 2c_{u}c_{v}\\
&= |a_{v}\times b_{v}|^{2} +  |a_{u}\times b_{u}|^{2} + |b_{v}|^{2}|a_{u}|^{2} + |b_{u}|^{2}|a_{v}|^{2}  - 2c_{u}c_{v}\\
&= |a_{v}\times b_{v}|^{2} +  |a_{u}\times b_{u}|^{2} +  |b_{v}|^{2}[|a_{u}\times \nu_{u}|^{2} + (a_{u}.\nu_{u})^{2}] + |b_{u}|^{2}[|a_{v}\times \nu_{v}|^{2}\\
&\hspace{8cm} + (a_{v}.\nu_{v})^{2}] - 2c_{u}c_{v}\\
&= |a_{v}\times b_{v}|^{2} +  |a_{u}\times b_{u}|^{2} +  |b_{v}|^{2}\frac{(c_{u})^{2}}{|b_{u}|^{2}} + |b_{v}|^{2}|a_{v}\times \nu_{v}|^{2} + |b_{u}|^{2}\frac{(c_{v})^{2}}{|b_{v}|^{2}}\\
&\hspace{6cm} + |b_{u}|^{2}|a_{u}\times \nu_{u}|^{2} - 2c_{u}|\frac{b_{v}}{b_{u}}| c_{v}|\frac{b_{u}}{b_{v}}|\\
&= |a_{v}\times b_{v}|^{2} +  |a_{u}\times b_{u}|^{2} + |b_{v}|^{2}|a_{u}\times \nu_{u}|^{2} + |b_{u}|^{2}|a_{v}\times \nu_{v}|^{2}     + \left(\frac{|b_{v}|c_{u}}{|b_{u}|} - \frac{|b_{u}|c_{v}}{|b_{v}|}\right)^{2} \\
&\geq |a_{v}\times b_{v}|^{2} +  |a_{u}\times b_{u}|^{2}  + |b_{v}|^{2}|a_{u}\times \nu_{u}|^{2} + |b_{u}|^{2}|a_{v}\times \nu_{v}|^{2}.
\end{align*}}
From the above inequality, we have
\begin{align*}
D &= (|b_{u}|^{2} + |b_{v}|^{2})\bigg[\bigg(\chi(\tau) + \frac{c_{v} + c_{u}}{|b_{u}|^{2} + |b_{v}|^{2}}\bigg)^{2} + \frac{(|b_{v}|^{2} + |b_{u}|^{2})(|a_{v}|^{2} + |a_{u}|^{2}) - (c_{v} + c_{u})^{2}}{(|b_{u}|^{2} + |b_{v}|^{2})^{2}}\bigg]\\
&\geq  \frac{(|b_{v}|^{2} + |b_{u}|^{2})(|a_{v}|^{2} + |a_{u}|^{2}) - (c_{v} + c_{u})^{2}}{|b_{v}|^{2} + |b_{u}|^{2}}\\
&\geq \frac{ |a_{v}\times b_{v}|^{2} +  |a_{u}\times b_{u}|^{2} +|b_{v}|^{2}|a_{u}\times \nu_{u}|^{2} + |b_{u}|^{2}|a_{v}\times \nu_{v}|^{2}}{|b_{v}|^{2} + |b_{u}|^{2}}\\
&= |a_{v}\times \nu_{v}|^{2} + |a_{u}\times \nu_{u}|^{2}.
\end{align*}
 We try to bound from below the terms $  |a_{v}\times \nu_{v}|^{2}$ and $ |a_{u}\times \nu_{u}|^{2}$.
\begin{align}\label{eq:4.13}
| b_{u}| &= |v\times u'|\nonumber\\
&\geq |\omega.v\times u'|\nonumber\\
&= |\omega.v\times (u + a \omega)|\nonumber\\
& = |\omega.v\times (u + a \omega)| = |\omega . p_{u}|.
\end{align}
Where for a given $x$, $ p_{x} $ is defined by $ p_{x} =  v\times x$.

Let's recall the following vector identity for three vectors $u$, $v$ and $w$
\begin{equation}\label{eq:4.14}
u\times (v\times w) = (u.w)v - (u.v)w
\end{equation}
Using (\ref{eq:4.14}), we have
\begin{align}\label{eq:4.15}
a_{v}\times b_{v} &= (x\times v')\times (v\times v')= -(v.x\times v')v'= v'. (v\times x)v' = (v'.p_{x})v'.
\end{align}
The same arguments as above yields to
\begin{equation}\label{eq:4.16}
a_{u}\times b_{u} = (u'.p_{x})u'
\end{equation}
In another hand, we have
\begin{equation}\label{eq:4.17}
b_{v} = v\times v' = v\times (v - a\omega)  = a p_{\omega}.
\end{equation}
This implies
\begin{equation*}
|a_{v}\times \nu_{v}| = \frac{|a_{v}\times b_{v}|}{|b_{v}|} = \frac{|(v'.p_{x})v'|}{|b_{v}|} = \frac{|v'.p_{x}||v'|}{|a||p_{\omega}|}.
\end{equation*}
In another hand, we have
\begin{equation*}
|v'| \geq |\omega\times v'|  = |\omega\times(v - a\omega)| = |\omega\times v| = |p_{\omega}|.
\end{equation*}
\begin{equation}\label{eq:4.18}
|a_{v}\times \nu_{v}| \geq \frac{|v'.p_{x}|}{|a|} = \frac{|(v - a\omega). x\times v |}{|a|} = |\omega.p_{x}|.
\end{equation}
Concerning $ a_{u}\times \nu_{u} $, using the relation $ a_{u}\times b_{u} = (u'.p_{x})u' $ we have
\begin{equation}\label{eq:4.19}
|a_{u}\times \nu_{u}|  \geq \frac{|u'.p_{x}|}{|v|}.
\end{equation}
(\ref{eq:4.16}) and (\ref{eq:4.19}) lead to
\begin{align}\label{eq:4.20}
D \geq |a_{v}\times \nu_{v}|^{2} + |a_{u}\times \nu_{u}|^{2} &\geq |\omega.p_{x}|^{2}  +  \frac{|u'.p_{x}|^{2}}{|v|^{2}} \geq |\omega.p_{x}|^{2}.
\end{align}
Let's now return to the estimate of the gain term. $ \int_{0}^{t}|Q_{g}^{\#}(f, f)(\tau, x, v)|d\tau $ enjoys
\begin{align}\label{eq:4.21}
&\int_{0}^{t}|Q_{g}^{\#}(f, f)(\tau, x, v)|d\tau \nonumber\\
&\leq c e^{-|v|^{2}}\|f^{\#}\|^{2}\int_{0}^{t}R^{-3}d\tau \int_{S^{2}_{+}\times \mathbb{R}^{3}} d\omega du v_{\phi}\sigma(g, \omega)e^{-|u|^{2}}e^{-|\omega.p_{x}|^{2}} \nonumber\\
&\leq \!\!c e^{-|v|^{2}}\|f^{\#}\|^{2}\!\!\int_{0}^{t}\!\!R^{-3}d\tau\!\! \int_{\mathbb{R}^{3}}du \frac{g\sqrt{s}}{v^{0}u^{0}}(1 + g^{-b})e^{-|u|^{2}}\!\!\int_{S^{2}_{+}}\!\!\sigma_{0}(\omega)e^{-|\omega.p_{x}|^{2}} d\omega.
\end{align}

- If $ |p_{x}| = |x\times v| \geq 1 $, under assumption (\ref{eq:2.11}), the integral with respect to $ \omega $ is controlled as follows
\begin{equation}\label{eq:4.22}
\int_{S^{2}_{+}}\sigma_{0}(\omega)e^{-|\omega.p_{x}|^{2}}d\omega \leq e^{-|p_{x}|^{2}} = e^{-|v\times x|^{2}}.
\end{equation}

- If $ |p_{x}| = |x\times v| < 1 $, we have
\begin{equation}\label{eq:4.23}
\int_{S^{2}_{+}}\sigma_{0}(\omega)e^{-|\omega.p_{x}|^{2}}d\omega \leq c \int_{S^{2}_{+}} d\omega \leq ce^{-|v\times x|^{2}}.
\end{equation}
(\ref{eq:4.22}) and (\ref{eq:2.23}) yield to
\begin{equation}\label{eq:4.24}
\int_{0}^{t}|Q_{g}^{\#}(f, f)(\tau, x, v)|d\tau \leq c\rho(x, v)^{-1}\|f^{\#}\|\int_{\mathbb{R}^{3}} \frac{g\sqrt{s}}{v^{0}u^{0}}(1 + g^{-b})e^{- |u|^{2}}du.
\end{equation}
The last term is controlled following the same arguments as for the loss term.
\end{proof}

We can now give the proof of our main result.
\begin{proof} If $ \|f_{0}\| \leq r/2 $ and $ f \in M_{r} $, then
\begin{align}\label{eq:4.26}
|\Gamma f^{\#}| &\leq \rho(x, v)^{-1}\|f_{0}\| + c\rho(x, v)^{-1}\|f^{\#}\|^{2}\leq  \rho(x, v)^{-1}[\frac{r}{2} + cr^{2}].
\end{align}
Thus, if $ \frac{r}{2} + cr^{2} \leq r $, i.e $ r \leq \frac{1}{2c} $, $ \Gamma $ maps $ M_{r} $ into itself.

On the other hand, using the bilinearity of Q, we prove that $ \Gamma $ is a contraction. In fact
If $ \|f_{0}\| \leq r/2 $ and $ f \in M_{r} $, then
\begin{equation}\label{eq:4.27}
|\Gamma f^{\#} - \Gamma g^{\#}| \leq c\rho(x, v)^{-1}(\|f\| + \|g\|)\|f - g\| \leq 2cR \rho(x, v)^{-1}\|f - g\|.
\end{equation}
The desired result is obtained if $ r \leq \frac{1}{2c} $.
\end{proof}

\textbf{Conclusion}: We have studied the inhomogeneous relativistic Boltzmann equation in RW space-time, previous studies were carried out for the spatially homogeneous case (see \cite{lee2,lee3,takou1}). We prove the global existence of mild solutions in a suitable weighted space.

\end{document}